\documentclass[10pt,usletter]{article}

\usepackage[totalwidth=450pt,totalheight=640pt]{geometry}

\usepackage{cite}
\usepackage{color}
\usepackage[linesnumbered,boxed]{algorithm2e}
\usepackage{graphicx}
\usepackage{latexsym}
\usepackage{amsmath}
\usepackage{amssymb}
\usepackage{amsthm}
\usepackage{enumerate}

\newtheorem{THE}{Theorem}

\newtheorem{COR}{Corollary}
\newtheorem{LEM}{Lemma}
\newtheorem{CLM}{Claim}

\newtheorem{problem}{Problem}

\newcommand{\naturals}{\mathbb{N}}

\newcommand{\cc}[1]{{\mbox{\textnormal{\textbf{#1}}}}}  

\newcommand{\NP}{\cc{NP}}
\newcommand{\paraNP}{\cc{para-NP}}

\newcommand{\FPT}{\cc{FPT}}
\newcommand{\Weft}{{\cc{W}}}
\newcommand{\W}[1]{{\Weft}{{[#1]}}}

\newcommand{\pn}[1]{\textsc{#1}}

\newcommand{\tuple}[1]{\langle{#1}\rangle}  

\newcommand{\instance}[1]{{\mathbb{#1}}}
\newcommand{\insti}{\instance{I}}

\newcommand{\hy}{\hbox{-}\nobreak\hskip0pt}

\newcommand{\SB}{\{\,}%
\newcommand{\SM}{\;{|}\;}%
\newcommand{\SE}{\,\}}%

\newcommand{\es}{;}
\newcommand{\ies}{-}
\newcommand{\gs}{\#}

\newcommand{\vertt}{\textbf{vt}}
\newcommand{\rep}{\textbf{rp}}
\newcommand{\enu}{\textbf{enu}}

\newcommand{\edges}{\textbf{e}}

\newcommand{\tij}{\textbf{t}}
\newcommand{\tvj}{\textbf{t}}
\newcommand{\tv}{\textbf{t}}

\newcommand{\pij}{\textbf{p}}
\newcommand{\pijs}{\textbf{pe}}

\begin{document}

\title{A Parameterized Study of Maximum Generalized Pattern Matching Problems}

\author{Sebastian Ordyniak\thanks{Research funded by 
      Employment of Newly Graduated Doctors of Science for
      Scientific Excellence (CZ.1.07/2.3.00/30.0009).} \and Alexandru Popa\\[0.1cm]
    Faculty of Informatics, Masaryk University, Brno, Czech Republic,\\ 
    \texttt{sordyniak@gmail.com, popa@fi.muni.cz}}

\maketitle
\date{ }

\begin{abstract} 
  The generalized function matching (GFM) problem  
  has been intensively studied starting with [Ehrenfeucht and Rozenberg, 1979].
  Given a pattern p and a text t, the goal is to find
  a mapping from the letters of p to non-empty substrings of t,
  such that applying the mapping to p results in t. 
  Very recently, the problem has been investigated within the framework
  of parameterized complexity [Fernau, Schmid, and Villanger, 2013].

  In this paper we study the parameterized complexity of
  the optimization variant of GFM (called Max-GFM), which has been introduced in
  [Amir and Nor, 2007]. Here, one is 
  allowed to replace some of the pattern letters with some special
  symbols ``?'', termed wildcards or don't cares, which can be mapped to
  an arbitrary substring of the text. The goal is to minimize the number
  of wildcards used. 

  We give a complete classification of the parameterized complexity 
  of Max-GFM and its variants under a wide range of parameterizations,
  such as, the number of occurrences of a letter in the text,
  the size of the text alphabet, the number of occurrences of a letter in the pattern, the size of the pattern
  alphabet, the maximum length of a string matched to any pattern letter,
  the number of wildcards and the maximum size of a string
  that a wildcard can be mapped to.
\end{abstract}

\section{Introduction}

In the generalized function matching problem one is given a text $t$ and a pattern $p$ and the goal is to decide whether there is a match between $p$ and $t$, where a
single letter of the pattern is allowed to match multiple
letters of the text (we say that $p$ GF-matches $t$). For example, if the text is $t = xyyx$ and the
pattern is $p = aba$, then a generalized function match (on short, GF-match)
is $a \rightarrow x, b \rightarrow yy$, but if $t = xyyz$ and $p =
aba$, then there is no GF-match. If, moreover, the matching is required to be injective, then we term the problem generalized parameterzied matching (GPM). In~\cite{AmirN07}, Amir and Nor describe applications of GFM in various areas such as software engineering, image searching, DNA analysis, poetry and music analysis, or author validation.  GFM is also related to areas such as (un-)avoidable patterns~\cite{Salomaa94}, word equations~\cite{MateescuS94} and the ambiguity of morphisms~\cite{frey05}.

GFM has a long history starting from 1979. Ehrenfeucht and Rozenberg~\cite{Ehrenfreucht1979} show that GFM is NP-complete. Independently, Angluin~\cite{Angluin79, Angluin80} studies a more general variant of GFM where the pattern may  contain also letters of the text alphabet. Angluin's paper received a lot of attention, especially in the learning theory community~\cite{Shinohara83,Reidenbach06,Reidenbach08} (see~\cite{Shinohara08} for a survey) but also in many other areas. 

Recently, a systematic study of the classical complexity of a number
of variants of GFM and GPM under various restrictions has been
carried out~\cite{FernauSchmid13,Schmid13}. It was
shown that GFM and GPM remain \NP\hy complete for many natural
restrictions. Moreover, the study of GFM and its variants within
the framework of parameterized complexity has recently been initiated~\cite{FernauSchmidVillanger13}.

In this paper we study the parameterized complexity of the
optimization variant of GFM (called Max-GFM) and its variants, where
one is allowed to replace some of the pattern letters with some special
symbols ``?'', termed wildcards or don't cares, which can be mapped to
an arbitrary substring of the text. The goal is to minimize the number
of wildcards used. The problem was first introduced to the pattern
matching community by Amir and Nor~\cite{AmirN07}. They show that if the pattern
alphabet has constant size, then a polynomial algorithm can be found,
but that the problem is \NP\hy complete otherwise. Then,
in~\cite{CliffordHPS09}, it is shown the \NP\hy hardness of the GFM (without
wildcards) and the \NP\hy hardness of the GFM when the function $f$ is
required to be an injection (named GPM). More specifically, 
 GFM is \NP\hy hard even if the text alphabet is binary and each
letter of the pattern is allowed to map to at most two letters
of the text~\cite{CliffordHPS09}. In the same paper it is given
a $\sqrt{OPT}$ approximation algorithm for the optimization
variant of GFM where the goal is to search for a pattern $p'$ that
GF-matches $t$ and has the smallest Hamming distance to $p$. In~\cite{CliffordP10} the
optimization versions of GFM and GPM are proved to be {\bf APX}-hard.



\paragraph*{\sc Our results}


Before we discuss our results, we give  formal definitions of the
problems. In the following let $t$ be a text over an alphabet
$\Sigma_t$ and let $p=p_1\dots p_m$
be a pattern over an alphabet $\Sigma_p$. We say that $p$
\emph{GF-matches} $t$ if there is a function $f : \Sigma_p \rightarrow
\Sigma_t^+$ such that $f(p_1)\dots f(p_m)=t$. To improve the
presentation we will sometimes abuse notation by writing $f(p)$
instead of $f(p_1)\dots f(p_m)$.
%
%
Let $k$ be a natural number. We say that a pattern $p$
\emph{$k$-GF-matches} $t$ if there is a text $p'$ over alphabet
$\Sigma_p \cup \{?_1, \dots,?_k\}$ of Hamming distance
at most $k$ from $p$ such that $p'$ GF-matches $t$.
\begin{problem}[\pn{Maximum Generalized Function Matching}] Given a
  text $t$, a pattern $p$, and an integer $k$, decide whether $p$ $k$-GF-matches $t$.
\end{problem}
The Max-GFM can be seen as the optimization variant of GFM in which we want
to replace some of the pattern letters with special wildcard
symbols, i.e., the symbols $?_1,\dots,?_k$, which can be mapped to any
non-empty substring of the text.

We  also study the Max-GPM
problem. The only difference between Max-GPM and Max-GFM is that for Max-GPM the function $f$ is
required to be injective. The notions of GP-matching and
$k$-GP-matching are defined in the natural way, e.g., we say a pattern
$p$ \emph{GP-matches} a text $t$ if $p$ \emph{GF-matches} $t$ using an
injective function. 

In this paper we study the parameterized complexity of the two
problems using a wide range of parameters: maximum number of occurrences of a letter in the text $\# \Sigma_t$,
maximum number of occurrences of a letter in the pattern $\# \Sigma_p$, size of the text alphabet
$|\Sigma_t|$, size of the pattern alphabet $|\Sigma_p|$, the maximum
length of a substring of the text that a letter of the pattern alphabet can be
mapped to (i.e., $\max_i | f(p_i) |$), the number of wildcard letters $\# ?$, and the maximum
length of a substring of the text that a wildcard can be
mapped to, denoted by $\max | f(?) |$.

Our results are summarized in Table~\ref{table:tab2}. We verified
the completeness of our results using a simple computer program.
In particular, the program checks for every of the $128$ possible
combinations of parameters $\mathcal{C}$ that the table contains either: i) a
superset of $\mathcal{C}$ under which Max-GFM/GPM is hard (and
thus, Max-GFM/GPM is hard if parameterized by $\mathcal{C}$); or ii) a
subset of $\mathcal{C}$ for which Max-GFM/GPM is fpt (and then
we have an fpt result for the set of parameters $\mathcal{C}$). Since
some of our results do not hold for both Max-GFM and Max-GPM, we carried out two separate checks, one for Max-GFM and one for Max-GPM.

\begin{table}[ht]
  \center
  \begin{tabular}{c|c|c|c|c|c|c|c}
    $\#\Sigma_t$ & $|\Sigma_t|$ & $\#\Sigma_p$ & $|\Sigma_p|$ & $\max_i |f(p_i)|$ & $\# ?$ & $\max |f(?)|$ & Complexity  \\
    \hline
    par & par & -- & -- & -- & -- & -- & \FPT\ (Cor.~\ref{cor:fpt-occt-sigt}) \\
    -- & par & -- & par & par & -- & -- & \FPT\ (Th.~\ref{the:fpt-max-gfm-text}) \\
    -- & par & -- & -- & par & -- & --  & \FPT\ only GPM (Cor.~\ref{cor:fpt-max-gpm}) \\
    -- & -- & par & par & par & -- & par & \FPT\ (Cor.~\ref{cor:fpt-max-gfm-wild}) \\
    -- & -- & -- & par & par & par & par & \FPT\ (Th.~\ref{the:fpt-max-gfm-wild}) \\    
    par & -- & par & par & par & par & -- & \W{1}\hy h (Th.~\ref{the:GFM-mobile-2}) \\
    par & -- & par & par & -- & par & par & \W{1}\hy h (Th.~\ref{the:GFM-mobile-1}) \\
    par & -- & par & -- & par & par & par & \W{1}\hy h (Th.~\ref{the:GFM-occt-max}) \\
    -- & par & par & par & -- & par & par & \W{1}\hy\ h (\cite[Th. 2.]{FernauSchmidVillanger13}) \\
    -- & -- & par & par & par & par & -- & \W{1}\hy\ h (Th.~\ref{the:GFM-questionmarksize-hard}) \\
    -- & -- & -- & par & par & -- & par & \W{1}\hy\ h (Th.~\ref{the:GFM-questionmark-hard}) \\
    -- & par & par & -- & par & par & par & \paraNP\hy h (\cite[Cor. 1]{AmirN07}), \\
    -- & par & par & -- & par & -- & -- & \paraNP\hy h only GFM~\cite{FernauSchmid13} \\
    -- & -- & par & -- & par & -- & -- & \paraNP\hy h only GPM~\cite{FernauSchmid13} \\
%
  \end{tabular}
  \caption{Parameterized Complexity of Max-GFM and Max-GPM .}
  \label{table:tab2}
\end{table}
%

The paper is organized as follows. In Section~\ref{sec:prel} we give preliminaries, in Section~\ref{sec:fpt} we present our fixed-parameter algorithms and in Section~\ref{sec:hardness} we show our hardness results.

\section{Preliminaries}
\label{sec:prel}

We define the basic notions of Parameterized Complexity and
refer to other sources~\cite{DowneyFellows99,FlumGrohe06} 
for an in-depth treatment. 
A \emph{parameterized problem} is a set of pairs 
$\tuple{\insti,k}$,
the \emph{instances}, where $\insti$ is the main part and $k$ 
the \emph{parameter}. The parameter is usually a non-negative integer.
A parameterized problem is \emph{fixed-parameter tractable (fpt)} if
there exists an algorithm that solves any instance $\tuple{\insti,k}$ of
size $n$ in time $f(k)n^{c}$ where $f$ is an arbitrary computable
function and $c$ is a constant independent of both $n$ and $k$. 
\FPT\ is the class of all fixed-parameter
tractable decision problems. Because we focus on fixed-parameter
tractability of a problem we will sometimes use the notation $O^*$ to
suppress exact polynomial dependencies, i.e.,
a problem with input size $n$ and parameter $k$ can be solved in 
time $O^*(f(k))$ if it can be solved in time $O(f(k)n^c)$ for some
constant $c$. 

Parameterized complexity offers a completeness theory, similar
to the theory of NP-completeness, that allows the accumulation of
strong theoretical evidence that some parameterized problems
are not fixed-parameter tractable. This theory is based on a
hierarchy of complexity classes
$\FPT \subseteq \W{1} \subseteq \W{2} \subseteq \W{3} \subseteq \cdots$
where all inclusions are believed to be strict. 
An \emph{fpt-reduction} from a parameterized problem $P$ to a
parameterized problem $Q$
is a mapping $R$ from instances of $P$ to instances of $Q$ such
that
(i)~$\tuple{\insti,k}$ is a {\sc Yes}-instance of $P$ if and only if $\tuple{\insti',k'}=R(\insti,k)$
is a {\sc Yes}-instance of $Q$,
(ii)~there is a computable function $g$ such that $k' \leq g(k)$, and
(iii) there is a computable function $f$ and a constant $c$ such that $R$ can be
computed in time $O(f(k) \cdot n^c)$, where $n$ denotes the size
of $\tuple{\insti,k}$.

For our hardness results we will often reduce from the following
problem, which is well-known to be \W{1}\hy
complete~\cite{Pietrzak03}.
\begin{quote}
  \noindent
  \pn{Multicolored Clique}\\
  \noindent  
  \emph{Instance:} A $k$-partite graph $G=\tuple{V,E}$ with 
  a partition $V_1,\dots,V_k$ of $V$.
  \\
  \noindent  
  \emph{Parameter:} The integer $k$.\\  
  \noindent
  \emph{Question:} Are there nodes $v_1,\dots,v_k$ such that
  $v_i\in V_i$ and $\{v_i,v_j\}\in E$ for all $i$ and $j$ with $1 \leq
  i < j \leq k$ 
  (i.e. the subgraph of $G$ induced by $\{v_1,\dots,v_k\}$ is a clique
  of size $k$)?
\end{quote}
For our hardness proofs we will often make the additional assumptions
that (1) $|V_i|=|V_j|$ for every $i$ and $j$ with $1 \leq i < j \leq
k$ and (2) $|E_{i,j}|=|E_{r,s}|$ for every $i$, $j$, $r$, and $s$ with
$1 \leq i < j \leq k$ and $1 \leq r < s \leq k$, where $E_{i,j}=\SB
\{u,v\} \in E \SM u \in V_i \textup{ and } v \in V_j \SE$ for every
$i$ and $j$ as before. To see that \pn{Multicolored Clique}
remains $\W{1}$\hy hard under these additional restrictions we can
reduce from \pn{Multicolored Clique} to its more restricted version
using a simple padding construction as
follows. Given an instance $\tuple{G,k}$ of \pn{Multicolored Clique}
we construct an instance of its more restricted version by adding 
edges (whose endpoints are new vertices) between parts
(i.e. $V_1,\dotsc,V_k$) that do not already have the maximum number of
edges between them and then adding isolated vertices to parts that do
not already have the maximum number of vertices.

Even stronger evidence that a parameterized problem is not
fixed-parameter tractable can be obtained by showing that the problem 
remains \NP\hy complete even if the parameter is a constant. The class of
these problems is called \paraNP.

A \emph{square} is a string consisting of two copies of the same
(non-empty) string. We say that a string is \emph{square-free} if it
does not contain a square as a substring.

\section{Fixed-parameter Tractable Variants}
\label{sec:fpt}

In this section we show our fixed-parameter tractability results 
for Max-GFM and Max-GPM. In particular, we show that Max-GFM and
Max-GPM are fixed-parameter tractable parameterized by $|\Sigma_t|$, $|\Sigma_p|$, and $\max_i |f(p_i)|$, and
also parameterized by $\# ?$, $\max|f(?)|$,
$|\Sigma_p|$, and $\max_i |f(p_i)|$. We start by showing
fixed-parameter tractability for the parameters $|\Sigma_t|$,
$|\Sigma_p|$, and $\max_i |f(p_i)|$. We need the following lemma.
\begin{LEM}\label{lem:dyn}
  Given a pattern $p=p_1\dots p_m$ over an alphabet $\Sigma_p$, a
  text $t=t_1\dots t_n$ over an alphabet $\Sigma_t$, a natural number $q$, and a
  function $f : \Sigma_p \rightarrow \Sigma_t^+$, then there is a
  polynomial time algorithm deciding whether $p$ $q$-GF/GP-matches $t$
  using the function $f$.
\end{LEM}
\begin{proof}
  If we are asked whether $p$ $q$-GP-matches
  $t$ and $f$ is not injective, then we obviously provide a negative answer. 
  Otherwise, we use a dynamic programming algorithm that is similar in spirit to
  an algorithm in~\cite{CliffordHPS09}. Let $\Sigma_p = \{a_1, \dots
  a_k\}$. For every $0 \leq i \leq j \leq n$, we define the function
  $g(i,j)$ to be the Hamming GFM/GPM-similarity (i.e., $m$ minus the minimum number of wildcards needed)
  between $t_1t_2 \dots t_j$ and $p_1p_2 \dots p_i$. 
  Then, we obtain the Hamming GFM/GPM-similarity between $p$ and $t$ as
  $g(m,n)$. Consequently, if $m-g(m,n)>q$, we return \textsc{No},
  otherwise we return \textsc{Yes}.

  We now show how to recursively compute $g(i,j)$. If $i=0$, we set $g(i,j)=0$ 
  and if $i \leq j$, we set:

$$g(i,j) = \max_{1 \leq k \leq j} \{ g(i-1,j-k) + I(t_{j-k+1} \dots
  t_j, f(p_i) \}$$

  \noindent where $I(s_1, s_2)$ is $1$ if the strings $s_1$, and $s_2$
  are the same, and $0$ otherwise.
  
  We must first show that the dynamic programming procedure computes the
  right function and then that it runs in polynomial time. We can see
  immediately that $g(0,i) = 0$ for all $i$ because in this case the
  pattern is empty. The recursion step of $g(i,j)$ has two cases:
  If $t_{j-|f(p_i)|+1} \dots t_j = f(p_i)$, then it is possible to map $p_i$ to $f(p_i)$, 
  and we can increase the number of mapped letters by one. Otherwise,
  we cannot increase the Hamming GFM/GPM-similarity. However, we
  know that $p_i$ has to be set to a wildcard and therefore we find
  the maximum of the previous results for different length
  substrings that the wildcard maps to. 
  
  It is straightforward to check that $g(i,j)$ can be computed in
  cubic time.
\end{proof}
\begin{THE}\label{the:fpt-max-gfm-text}\sloppypar
  Max-GFM and Max-GPM parameterized by 
  $|\Sigma_t|$, $|\Sigma_p|$, and $\max_i |f(p_i)|$ are fixed-parameter
  tractable.
\end{THE}
\begin{proof}
  Let $p$, $t$, and $q$ be an instance of Max-GFM or
  Max-GPM, respectively. The pattern $p$ $q$-GF/GP-matches $t$ if and only if there is a function 
  $f : \Sigma_p \rightarrow \Sigma_t^+$ such that $p$ $q$-GF/GP-matches $t$
  using $f$. Hence, to solve
  Max-GFM/Max-GPM, it is sufficient
  to apply the algorithm from Lemma~\ref{lem:dyn}
  to every function $f : \Sigma_p \rightarrow \Sigma_t^+$ that could
  possible constitute to a $q$-GF/GP-matching from
  $p$ to $t$. Because there
  are at most ${(|\Sigma_t|)^{\max_i|f(p_i)|}}^{|\Sigma_p|}$ such
  functions $f$ and the algorithm from Lemma~\ref{lem:dyn} runs in
  polynomial time, the running time of this algorithm is
  $O^*({(|\Sigma_t|)^{\max_i|f(p_i)|}}^{|\Sigma_p|})$, and hence
    fixed-parameter tractable in $|\Sigma_t|$, $|\Sigma_p|$, and
    $\max_i |f(p_i)|$.
\end{proof}

Because in the case of Max-GPM it holds that if $|\Sigma_t|$ and
$\max_i |f(p_i)|$ is bounded then also $\Sigma_p$ is bounded by
$|\Sigma_t|^{\max_i |f(p_i)|}$, we obtain the following corollary.
\begin{COR}\label{cor:fpt-max-gpm}\sloppypar
  Max-GPM parameterized by 
  $|\Sigma_t|$ and $\max_i |f(p_i)|$ is fixed-parameter
  tractable.
\end{COR}

We continue by showing our second tractability result for the
parameters $|\Sigma_p|$, $\max_i |f(p_i)|$, $\# ?$, and $\max|f(?)|$.
\begin{THE}\sloppypar
  \label{the:fpt-max-gfm-wild}
  Max-GFM and Max-GPM parameterized by
  $|\Sigma_p|$, $\max_i|f(p_i)|$, $\# ?$, $\max|f(?)|$,  are
  fixed-parameter tractable.
\end{THE}

\begin{proof}
  Let $p$, $t$, and $q$ be an instance of Max-GFM or Max-GPM,
  respectively.

  Observe that if we could go over all possible functions $f :
  \Sigma_p \rightarrow \Sigma_t^+$ that could
  possible constitute to a $q$-GF/GP-matching from
  $p$ to $t$, then we could again apply Lemma~\ref{lem:dyn} as we did in the
  proof of Theorem~\ref{the:fpt-max-gfm-text}. Unfortunately, because
  $|\Sigma_t|$ is not a parameter, the number of these functions cannot
  be bounded as easily any more. However, as we will show next it is
  still possible to bound the number of possible functions solely in
  terms of the parameters. In particular, we will show that the number
  of possible substrings of $t$ that any letter of the pattern
  alphabet can be mapped to is bounded by a function of the
  parameters. Because also $|\Sigma_p|$ is a parameter this immediately
  implies a bound (only in terms of the given parameters) on the total
  number of these functions.

  Let $c \in \Sigma_p$ and consider any $q$-GF/GP-matching from $p$ to $t$, i.e., a text $p'=p_1'\dots p_m'$ of
  Hamming distance at most $q$ to $p$ and a
  function $f : \Sigma_p \cup \{?_1,\dots,?_q\} \rightarrow \Sigma_t^+$
  such that $f(p_1')\dots f(p_m')=t$.
  Then either $c$ does not occur in $p'$ or $c$ occurs in $p'$. In the first case
  we can assign to $c$ any non-empty substring over the alphabet
  $\Sigma_t$
  (in the case of Max-GPM one additionally has to ensure that the
  non-empty substrings over $\Sigma_t$ that one chooses for distinct
  letters in $\Sigma_p$ are distinct).
  In the second case let $p_i'$ for some $i$ with $1 \leq i \leq m$ be
  the first
  occurrence of $c$ in $p'$, let $\overline{p}_{i-1}'=p_1'\dots
  p_{i-1}'$, and let $\overline{p}_{i-1}=p_1\dots p_{i-1}$. 
  Furthermore, for every $b \in \Sigma_p \cup \{?_1,\dots,?_q\}$ 
  and $w \in (\Sigma_p \cup \{?_1,\dots,?_q\})^*$, we denote by
  $\#(b,w)$ the number of times $b$ occurs in $w$. Then $f(c)=t_{c_s+1}\dots t_{c_s+|f(c)|}$
  where $c_s=\sum_{j=1}^{i-1}|f(p_j')|$, which implies that the value of $f(c)$
  is fully determined by $c_s$ and $|f(c)|$. Because the number of
  possible values for $|f(c)|$ is trivially bounded by the parameters
  (it is bounded by $\max_i|f(p_i)|$), it remains to show that also
  $c_s$ is bounded by the given parameters.

  Because
  $c_s=\sum_{j=1}^{i-1}|f(p_j')|=(\sum_{b \in
    \Sigma_p \cup \{?_1,\dots,?_q\}}\#(b,\overline{p}_{i-1}')|f(b)|)$,
  we obtain that the value of $c_s$ is fully determined by the values of
  $\#(b,\overline{p}_{i-1}')$ and $|f(b)|$ for every $b \in \Sigma_p \cup
  \{?_1,\dots,?_q\}$. For every $? \in \{?_1,\dots,?_q\}$ there are at
  most $2$ possible values for $\#(?,\overline{p}_{i-1}')$ (namely $0$
  and $1$) and there are at
  most $\max |f(?)|$ possible values for $|f(?)|$. Similarly, for every
  $b \in \Sigma_p$ there are at most $q+1$ possible values for
  $\#(b,\overline{p}_{i-1}')$ (the values $\#(b,\overline{p}_{i-1})-q,
  \dots, \#(b,\overline{p}_{i-1})$) and there are at most
  $\max_i|f(p_i)|$ possible values for $|f(b)|$. Hence, the number of
  possible values for $c_s$ is bounded in terms of the parameters, as
  required.
\end{proof}

Since $|\Sigma_p|$ and $\#\Sigma_p$ together bound $\# ?$, we obtain the following
corollary.
\begin{COR}\sloppypar
  \label{cor:fpt-max-gfm-wild}
  Max-GFM and Max-GPM parameterized by
  $\#\Sigma_p$, $|\Sigma_p|$, $\max_i|f(p_i)|$, and $\max|f(?)|$ are fixed-parameter tractable.
\end{COR}
Furthermore, because all considered parameters can be bounded in terms of the
parameters $\#\Sigma_t$ and $|\Sigma_t|$, we obtain the following
corollary as a consequence of any of our above fpt-results.
\begin{COR}\label{cor:fpt-occt-sigt}
  Max-GFM and Max-GPM parameterized by
  $\#\Sigma_t$ and $|\Sigma_t|$ are fixed-parameter tractable.
\end{COR}

\section{Hardness Results}
\label{sec:hardness}

In this subsection we give our hardness results for Max-GFM and
Max-GPM. 

\begin{THE}\label{the:GFM-questionmark-hard}
  Max-GFM and Max-GPM are \W{1}\hy hard parameterized by $|\Sigma_p|$, 
  $\max_i |f(p_i)|$, and $\max|f(?)|$  (even if $\max_i |f(p_i)|=1$
  and $\max|f(?)|=2$).
\end{THE}

We will show the theorem by a  parameterized reduction from \textsc{Multicolored
 Clique}. To simplify the proof we will reduce to the variant of
Max-GFM and Max-GPM, where we are allowed to map wildcards to the empty
string. It is however straightforward to adapt the proof to the
original versions of Max-GFM and Max-GPM. Hence, in the following,
whenever we refer to Max-GFM and Max-GPM, we mean the version of
Max-GFM and Max-GPM, where wildcards can be mapped to the empty string.

Let $G = (V, E)$ be a $k$-partite graph with partition
$V_1,\dots,V_k$ of $V$. Let $E_{i,j}=\SB \{u,v\} \in E \SM  u \in
V_i \textup{ and } v \in V_j \SE$ for every $i$ and $j$ with $1 \leq
i < j \leq k$. 
Again, as we stated in the preliminaries we can assume that $|V_i|=n$ and
$|E_{i,j}|=m$ for every $i$ and $j$ with $1 \leq i < j \leq k$.

Let $V_i=\{v_1^i,\dotsc,v_n^i\}$ and $E_{i,j}=\{e_1^{i,j},\dotsc,e_m^{i,j}\}$.
We construct a text $t$ and a pattern $p$ from $G$ and $k$ such that
$p$ $r$-GF/GP-matches $t$ with $r=\binom{k}{2}(8(m-1))$ if and only if $G$ has a $k$-clique. We set
$\Sigma_t=\{\es,\ies,\gs,\Box\} \cup \SB v_i^j \SM 1 \leq i \leq n
\textup{ and }1 \leq j \leq k \SE$ and 
$\Sigma_p=\{\es,\ies,\gs, \Box\} \cup \SB V_i \SM 1 \leq i \leq k \SE$.

For an edge $e \in E$ between $v_l^i$ and $v_k^j$ where $1 \leq i
< j \leq k$ and $1 \leq l,k \leq n$, we write $\vertt(e)$ to
denote the text $v_l^i \ies v_k^j$. For $l \in \Sigma_p \cup \Sigma_t$ and $i \in
\naturals$ we write $\rep(l,i)$ to denote the text consisting of
repeating the letter $l$ exactly $i$ times.
We first define a preliminary text $t'$ as follows.
\begin{quote}
 \begin{center}
   $\gs \es \vertt(e_1^{1,2}) \es \dotsb \es \vertt(e_m^{1,2}) \es
   \gs \dotsb \gs \es \vertt(e_1^{1,k}) \es \dotsb \es
   \vertt(e_m^{1,k}) \es$
   
   $\gs \es \vertt(e_1^{2,3}) \es \dotsb \es \vertt(e_m^{2,3}) \es
   \gs \dotsb \gs \es \vertt(e_1^{2,k}) \es \dotsb \es \vertt(e_m^{2,k})\es$
   
   $\dotsb$
   
   $\gs \es \vertt(e_1^{k-1,k}) \es \dotsb \es
   \vertt(e_m^{k-1,k}) \es \gs$
 \end{center}
\end{quote}
We also need to define a preliminary pattern $p'$ as follows.
\begin{quote}
 \begin{center}
   $\gs \rep(\Box,4(m-1)) \es V_1 \ies V_2 \es \rep(\Box,4(m-1)) \gs
   \dotsc$

   $\gs \rep(\Box,4(m-1)) \es V_1 \ies V_k \es \rep(\Box,4(m-1))$
   
   $\gs \rep(\Box,4(m-1)) \es V_2 \ies V_3 \es \rep(\Box,4(m-1)) \gs
   \dotsc$

   $\gs \rep(\Box,4(m-1)) \es V_2 \ies V_k \es \rep(\Box,4(m-1))$
   
   $\dotsb$
   
   $\gs \rep(\Box,4(m-1)) \es V_{k-1} \ies V_k \es \rep(\Box,4(m-1)) \gs$
 \end{center}
\end{quote}
We obtain $t$ from $t'$ and $p$
from $p'$ by preceding $t'$ and $p'$, respectively, with the
following text or pattern, respectively.
\begin{quote}
 \begin{center}
   $\rep(\Box,2r+1)\rep(\es,2r+1)\rep(\ies,2r+1)\rep(\gs,2r+1)$
 \end{center}
\end{quote}
This completes the construction of $t$ and $p$. Clearly, $t$ and $p$
can be constructed from $G$ and $k$ in fpt-time (even polynomial
time). Furthermore, $|\Sigma_p|=k+4$, as required. It remains to show
that $G$ has a $k$-clique if and only if $p$ $r$-GF/GP-matches $t$. 
\begin{LEM}\label{lem:GFM-questionmarksize-hard-1}
 If $G$ has a $k$-clique then $p$ $r$-GF/GP-matches $t$.
\end{LEM}
\begin{proof}
 Let $\{v_{h_1}^1, \dotsc,v_{h_k}^k\}$ be the vertices and 
 $\SB e^{i,j}_{h_{i,j}} \SM 1 \leq i < j \leq k \SE$ be the edges of
 a $k$-clique of $G$ with $1 \leq h_j \leq n$  and $1 \leq h_{i,j}
 \leq m$ for every $i$ and $j$ with $1 \leq i < j \leq k$. 

 The function $f$ that $r$-GF/GP-matching $p$ to $t$
 is defined as follows: 
 $f(\Box)=\Box$, $f(\es)=\es$, $f(\ies)=\ies$, $f(\gs)=\gs$, 
 $f(V_i)=v_{h_1}^i$, 
 for every $i$ with $1 \leq i \leq k$. 
 
 We put $r$ wildcards on the last $r$ occurrences of $\Box$ in
 $p$, i.e., every occurrence of $\Box$ that corresponds to an
 occurrence in $p'$. Then length of the text the wildcards are
 mapped to is determined as follows. For an edge $e^{i,j}_{h_{i,j}}$
 look at the ``block'' in $p$ that corresponds to the edge, i.e., the
 block:
 \begin{center}
   \begin{quote}
     $\gs \rep(\Box,4(m-1)) \es V_i \ies V_j \es \rep(\Box,4(m-1))$ 
   \end{quote}
 \end{center}
 The first $4(m-h_{i,j})$ occurrences of $\Box$ (in this block) are
 replaced with a
 wildcard which is mapped to a text of length $0$, the last
 $4(m-h_{i,j})$ occurrences of $\Box$ are replaced with a
 wildcard which is mapped to a text of length $2$, and all other
 occurrences of $\Box$ are replaced with a wildcard that is
 mapped to a text of length $1$. It is straightforward to check that 
 $f$ together with the mapping of the wildcards
 maps the pattern $p$ to the text $t$.
\end{proof}
For the reverse direction we need the following intermediate claims.
\begin{CLM}\label{clm:GFM-questionmarksize-hard-s21}
 For any function $f$ that $r$-GF/GP-matches $p$ to $t$ it holds that:
 $f(\Box)=\Box$, $f(\es)=\es$, $f(\ies)=\ies$, and $f(\gs)=\gs$.
\end{CLM}
\begin{proof}
 We show that $f(\Box)=\Box$ since the remaining cases are similar. Because
 the pattern $p$ starts with $2r+1$ repetitions of the letter $\Box$,
 it follows that at least $1$ of these occurrences of $\Box$
 is not replaced with a wildcard. Because every letter of $p$
 is replaced by at most $2$ letters of the text the first occurrence
 of $\Box$ that is not replaced by a wildcard is mapped to a
 letter of the text at position at most $2r$, i.e., a $\Box$.
 This concludes the proof of the claim.
\end{proof}
\begin{CLM}\label{clm:GFM-questionmarksize-hard-s22}
 Any $r$-GF/GP-matching of $p$ to $t$ replaces exactly the last $r$
 occurrences of $\Box$ in $p$ with wildcards.
\end{CLM}
\begin{proof}
 It follows from the previous claim that $f(\Box)=\Box$ for any
 function that $r$-GF/GP-matches $p$ to $t$. Because every letter of $p$
 that is not replaced with a wildcard is replaced with exactly
 $1$ letter from the text, it follows that the first occurrence of
 $\Box$ in $p$ that corresponds to an occurrence of $\Box$ in $p'$ is mapped
 to (if it is not replaced with a wildcard) to a letter of the
 text at position at least $7r+5$. However, since the text $t$ does
 not contain the letter $\Box$ after position $2r+1$, this occurrence
 of $\Box$ in $p$ (and all other occurrences of $\Box$ in $p$ that
 follow) has to be replaced with a wildcard. Since $p'$ contains
 exactly $r$ occurrences of $\Box$ the only letters of $p$ that
 are replaced with wildcards are these occurrences of $\Box$.
\end{proof}
\begin{LEM}\label{lem:GFM-questionmarksize-hard-2}
 If $p$ $r$-GF/GP-matches $t$ then $G$ has a $k$-clique.
\end{LEM}
\begin{proof}
 Let $f$ be a function that $r$-GF/GP-matches $p$ to $t$. Because of
 Claim~\ref{clm:GFM-questionmarksize-hard-s21}, it holds that
 $f(\Box)=\Box$, $f(\es)=\es$, $f(\ies)=\ies$, and
 $f(\gs)=\gs$. Furthermore, because of
 Claim~\ref{clm:GFM-questionmarksize-hard-s21} the only
 letters in $p$ that are replaced with wildcards are the last $r$
 occurrences of $\Box$ in $p$. Because the number of occurrences of
 the letter $\gs$ is the same in $t$ and $p$ each occurrence of $\gs$
 in $p$ has to be mapped to its corresponding occurrence in $t$. It
 follows that for every $i$ and $j$ with $1 \leq i < j \leq k$ the
 ``block''
 \begin{center}
   \begin{quote}
     $\rep(\Box,4(m-1)) \es V_i \ies V_j \es \rep(\Box,4(m-1))$ 
   \end{quote}
 \end{center}
 in $p$ has to be mapped to the corresponding ``block''
 \begin{center}
   \begin{quote}
     $\es \vertt(e_1^{i,j}) \es \dotsb \es \vertt(e_m^{i,j}) \es$
   \end{quote}
 \end{center}
 in $t$. Hence, the part $V_i \ies V_j$ has to be mapped to
 $\vertt(e_{l}^{i,j})$ for every $1 \leq l \leq m$. Consequently, the
 set $\SB f(V_i) \SM 1 \leq i \leq k \SE$ is a $k$-clique of
 $G$. 
\end{proof}
This concludes the proof of Theorem~\ref{the:GFM-questionmark-hard}.



\begin{THE}\label{the:GFM-mobile-2}
  Max-GFM and Max-GPM are \W{1}\hy hard parameterized by $\#\Sigma_t$,
  $\# \Sigma_p$, $|\Sigma_p|$, $\max_i |f(p_i)|$, and $\# ?$.
\end{THE}

We will show the theorem by a  parameterized reduction from \textsc{Multicolored
 Clique}. Let $G = (V, E)$ be a $k$-partite graph with partition
$V_1,\dots,V_k$ of $V$. Let $E_{i,j}=\SB \{u,v\} \in E \SM  u \in
V_i \textup{ and } v \in V_j \SE$ for every $i$ and $j$ with $1 \leq
i < j \leq k$. 
Again, as we stated in the preliminaries we can assume that $|V_i|=n$ and
$|E_{i,j}|=m$ for every $i$ and $j$ with $1 \leq i < j \leq k$.

Let $V_i=\{v_1^i,\dotsc,v_n^i\}$ and $E_{i,j}=\{e_1^{i,j},\dotsc,e_m^{i,j}\}$, and $k'=2\binom{k}{2}+k(k+2)$.
We construct a text $t$ over alphabet $\Sigma_t$ and a pattern $p$
over alphabet $\Sigma_p$ from $G$ and $k$ such that
$p$ $k'$-GF/GP-matches $t$ using a function $f$ with $\max_{p \in
 \Sigma_p}|f(p)|=1$ if and only if $G$ has a $k$-clique. 
The alphabet $\Sigma_t$ consists of:
\begin{itemize}
\item the letter $\gs$ (used as a separator);
\item the letter $+$ (used to forced the wildcards);
\item one letter $a_e$ for every $e \in E$ (representing the
 edges of $G$);
\item one letter $\gs_i$ for every $i$ with $1 \leq i \leq n$ (used as
 special separators that group edges from the same vertex);
\item the letters $l_{i,j}$, $r_{i,j}$, $l_i$, $r_i$ for every $i$ and $j$ with
 $1 \leq i < j \leq k$ (used as dummy letters to ensure injectivity
 for GPM);
\item the letter $d_e^v$ and $d^v$ for every $e \in E$ and $v \in V(G)$ with
 $v \in e$ (used as dummy letters to ensure injectivity for GPM).
\end{itemize}
We set $\Sigma_p=\{\gs,D\} \cup \SB E_{i,j} \SM 1 \leq i < j \leq k \SE$.

For a vertex $v \in V$ and $j$ with $1 \leq j
\leq k$ we denote by $E_j(v)$ the set of edges of $G$ that are
incident to $v$ and whose other endpoint is in $V_j$. 
Furthermore, for a vertex $v \in V(G)$, we write $\edges(v)$ to
denote the text $\textbf{el}(v,E_1(v))\dotsb \textbf{el}(v,E_k(v))d^v$,
where $\textbf{el}(v,E')$, for vertex $v$ and a set $E'$ of edges with $E'=\{e_1,\dotsc,e_l\}$, 
is the text $d_{e_1}^v e_{e_1} d_{e_2}^v e_{e_2} \dotsb d_{e_l}^v
a_{e_l}$.

We first define the following preliminary text and pattern
strings. Let $t_1$ be the text:
\begin{quote}
 \begin{center}
   $\gs l_{1,2} a_{e^{1,2}_1} \dotsb a_{e^{1,2}_m} r_{1,2} \gs \dotsb \gs l_{1,k} a_{e^{1,k}_1} \dotsb a_{e^{1,k}_m} r_{1,k}$
   
   $\gs l_{2,3} a_{e^{2,3}_1} \dotsb a_{e^{2,3}_m} r_{2,3} \gs \dotsb \gs l_{2,k} a_{e^{2,k}_1} \dotsb a_{e^{2,k}_m} r_{2,k}$
   
   $\dotsb$
   
   $\gs l_{k-1,k} a_{e^{k-1,k}_1} \dotsb a_{e^{k-1,k}_m} r_{k-1,k}$
 \end{center}
\end{quote}
Let $t_2$ be the text:
\begin{quote}
 \begin{center}
   $\gs l_1 \gs_1 \edges(v_1^1) \gs_1 \dotsb \gs_n \edges(v_n^1) \gs_n
   r_1$
   
   $\dotsb$
   
   $\gs l_k \gs_1 \edges(v_1^k) \gs_1 \dotsb \gs_n \edges(v_n^k)
   \gs_n r_k \gs$
 \end{center}
\end{quote}
Let $p_1$ be the pattern:
\begin{quote}
 \begin{center}
   
   $\gs D E_{1,2} D \gs \dotsc \gs D E_{1,k} D$
   
   $\gs D E_{2,3} D \gs \dotsc \gs D E_{2,k} D$
   
   $\dotsb$
   
   $\gs D E_{k-1,k} D$
 \end{center}
\end{quote}

\newcommand{\pih}{\mathbf{p}}

For $i$, $j$ with $1 \leq i,j \leq k$, let $I(i,j)$ be 
the letter $E_{i,j}$ if $i < j$, the letter $E_{j,i}$ if $i>j$ and the empty string if
$i=j$. We define $\pih(1)$ to be the
pattern:
\begin{quote}
 \begin{center}
   $A_1 D I(1,2) D I(1,3) \dotsb \dotsb D I(1,k) D A_1$
 \end{center}
\end{quote}
, we define $\pih(k)$ to be the
pattern:
\begin{quote}
 \begin{center}
   $A_k D I(k,1) D I(k,2) \dotsb \dotsb D I(k,k-1) D A_k$
 \end{center}
\end{quote}
, and for every $i$ with $1 < i < k$, we define
$\pih(i)$ to be the pattern:
\begin{quote}
 \begin{center}
   $A_i D I(i,1) D I(i,2) \dotsb
   D I(i,i-1) D I(i,i+1) \dotsb D I(i,k) D
   A_i$
 \end{center}
\end{quote}
Then $p_2$ is the pattern:
\begin{quote}
 \begin{center}

   $\gs L_1 \pih(1) R_1 \gs \dotsb \gs L_k \pih(k) R_k \gs$
 \end{center}
\end{quote}
Let $r=2(k'+1)$. For $l \in \Sigma_p \cup \Sigma_t$ and $i \in
\naturals$ we write $\rep(l,i)$ to denote the text consisting of
repeating the letter $l$ exactly $i$ times.
We also define $t_0$ to be the text $\gs\rep(+,r)$ and $p_0$ to be
the pattern $\gs\rep(D,r)$.
Then, $t$ is the concatenation of $t_0$, $t_1$ and $t_2$
and $p$ is a concatenation of $p_0$, $p_1$ and $p_2$.

This completes the construction of $t$ and $p$. Clearly, $t$ and $p$
can be constructed from $G$ and $k$ in fpt-time (even polynomial
time). Furthermore, $\# \Sigma_t=r$, 
$\# \Sigma_p=r+k'$,
$|\Sigma_p|=\binom{k}{2}+k+2$
and hence bounded by $k$,
as required. It remains to show that $G$ has a $k$-clique if and only
if $p$ $k'$-GF/GP-matches $t$ using a function $f$ with $\max_{p \in \Sigma_p}|f(p)|=1$. 

\begin{LEM}\label{lem:GFM-mobile-2-hard-1}
 If $G$ has a $k$-clique then $p$ $k'$-GF/GP-matches $t$ using a
 function $f$ with $\max_{p \in \Sigma_p}|f(p)|=1$.
\end{LEM}
\begin{proof}\sloppypar
 Let $\{v_{h_1}^1, \dotsc,v_{h_k}^k\}$ be the vertices and 
 $\SB e_{h_{i,j}}^{i,j} \SM 1 \leq i < j \leq k \SE$ be the edges of
 a $k$-clique of $G$ with $1 \leq h_i \leq n$ and $1 \leq h_{i,j}
 \leq m$ for every $i$ and $j$ with $1 \leq i < j \leq k$. 

 We put $k'$ wildcards on the last $k'$ occurrences of $D$ in $p$.
 The mapping of these wildcards is defined very similar to the
 mapping of the letters $L_{i,j}$, $R_{i,j}$, $L_i$, $R_i$, and
 $D_{i,j}$ in the proof of Lemma~\ref{lem:GFM-mobile-1-hard-1}
 and will not be repeated here. Using this mapping ensures that every
 wildcard is mapped to an non-empty substring of $t$ and no two
 wildcards are mapped to the same substring of $t$.
 
 We define the function $f$ that $k'$-GF/GP-matches $p$ to $t$ as follows:
 We set $f(\gs)=\gs$ and $f(D)=+$. Moreover, for every $i$ and
 $j$ with $1 \leq i < j \leq k$, we set
 $f(E_{i,j})=a_{e_{h_{i,j}}^{i,j}}$ and $f(A_i)=\gs_{i}$.

 It is straightforward to check that $f$ together with above mapping
 for the wildcards $k'$-GF/GP-matches $p$ to $t$.
\end{proof}

\begin{CLM}\label{clm:GFM-mobile-2-hard-1}
 Let $f$ be a function that $k'$-GF/GP-matches $p$ to $t$ with
 $\max_{p \in \Sigma_p}|f(p)|=1$, then:
 $f(\gs)=\gs$ and $f(D)=+$.
 Moreover, all wildcards have to be placed on all the $k'$
 occurrences of $D$ in $p'$.
\end{CLM}
\begin{proof}
 We first show that $f(D)=+$. Observe that the only squares in the
 string $t$ are contained in $t_0$ (recall the definition of square-free from
 Section~\ref{sec:prel}). It follows that every two consecutive
 occurrences of pattern letters in $p_0$ have to be mapped to a
 substring of $t_0$. Because there are $2(k'+1)$ occurrences of $D$
 in $p_0$ it follows that at least two consecutive occurrences of
 $D$ in $p_0$ are not replaced with wildcards and hence $D$ has to be
 mapped to a substring of $t_0$. Furthermore, since all occurrences
 of $D$ are at the end of $p_0$, we obtain that $D$ has to be mapped
 to $+$, as required. Because all occurrences of $D$ in $p'$ have to
 be mapped to substrings of the concatenation of $t_1$ and $t_2$, but
 these strings
 do not contain the
 letter $+$, it follows that all the $k'$ occurrences of $D$ in $p_1$
 and $p_2$
 have to be replaced by wildcards. Since we are only allowed to use
 at most $k'$ wildcards, this shows the second statement of the
 claim. Since no wildcards are used to replace letters in $p_0$ it
 now easily follows that $f(\gs)=\gs$.
\end{proof}

\begin{LEM}\label{lem:GFM-mobile-2-hard-1}
 If $p$ $k'$-GF/GP-matches $t$ using a function $f$ with $\max_{p \in
   \Sigma_p}|f(p)|=1$, then $G$ has a $k$-clique.
\end{LEM}
\begin{proof}
 Let $f$ be a function that $k'$-GF/GP-matches $p$ to $t$ with
 $\max_{p \in \Sigma_p}|f(p)|=1$. 
 Because of
 Claim~\ref{clm:GFM-mobile-2-hard-1}, we know that
 $f(\gs)=\gs$ and that no
 occurrence of $\gs$ in $p$ is replaced by a wildcard.
 Because $t$ and $p$ have the same number of occurrences of $\gs$,
 it follows that the $i$-th occurrences of $\gs$ in $p$ has to be
 mapped to the $i$-th occurrence of $\gs$ in $t$. We obtain that:
 \begin{enumerate}
 \item[(1)] For every $i$, $j$ with $1 \leq i < j \leq k$, the substring
   $D E_{i,j} D$ of $p$ has to be mapped to the substring
   $l_{i,j} a_{e^{i,j}_1} \dotsb a_{e^{i,j}_m} r_{i,j}$ of $t$.
 \item[(2)] For every $i$ with $1 \leq i \leq k$,  the substring
   $L_i \pih(i) R_i$ of $p$ has to be mapped to the substring
   $l_i \gs_1 \edges(v_1^i) \gs_1 \dotsb \gs_n \edges(v_n^i) \gs_n r_i$ of $t$.
 \end{enumerate}
 Because for every $i$ with $1 \leq i \leq k$ the letters $\gs_j$
 are the only letters that occur more than once in the substring 
 $l_i \gs_1 \edges(v_1^i) \gs_1 \dotsb \gs_n \edges(v_n^i) \gs_n r_i$
 of $t$, we obtain from (2) that $A_i$ has to be mapped to $\gs_j$ for some
 $j$ with $1 \leq j \leq n$. Consequently:
 \begin{enumerate}
 \item[(3)] for every $i$ with $1 \leq i \leq k$,  the substring
   $\pih(i)$ of $p$
   has to be mapped to a substring
   $\gs_j \edges(v_j^i) \gs_j$ of $t$ for some $j$ with $1 \leq j
   \leq n$.
 \end{enumerate}
 It follows from (1) that for every $i$, $j$ with $1 \leq i < j
 \leq k$, $f(E_{i,j})$ is mapped to an edge between $V_i$ and
 $V_j$. 
 Furthermore, because of (3) it follows that for every $i$
 with $1 \leq i \leq k$, it holds that the edges mapped to any
 $E_{l,r}$ with $1 \leq l < r \leq k$ such that $l=i$ or $r=i$ have
 the same endpoint in $V_i$. Hence, the set of edges mapped to 
 the letters $E_{i,j}$ for $1 \leq i <j \leq k$ form a $k$-clique
 of $G$.
\end{proof}
This concludes the proof of Theorem~\ref{the:GFM-mobile-2}.

\begin{THE}\label{the:GFM-occt-max}
  (Max-)GFM and (Max-)GPM are \W{1}\hy hard parameterized by $\#\Sigma_t$,
  $\# \Sigma_p$, $\max_i|f(p_i)|$, $\# ?$, and $\max |f(?)|$.
\end{THE}

We will show the above theorem by a  parameterized reduction from \textsc{Multicolored
 Clique}. Let $G = (V, E)$ be a $k$-partite graph with partition
$V_1,\dots,V_k$ of $V$. Let $E_{i,j}=\SB \{u,v\} \in E \SM  u \in
V_i \textup{ and } v \in V_j \SE$ for every $i$ and $j$ with $1 \leq
i < j \leq k$. 
Again, as we stated in the preliminaries we can assume that $|V_i|=n$ and
$|E_{i,j}|=m$ for every $i$ and $j$ with $1 \leq i < j \leq k$.

Let $V_i=\{v_1^i,\dotsc,v_n^i\}$ and
$E_{i,j}=\{e_1^{i,j},\dotsc,e_m^{i,j}\}$. 
For a vertex $v \in V_i$ and $j$ with $1 \leq j \leq k$ we denote by
$E_j(v)$ the set of edges of $G$ that are incident to $v$ and whose
other endpoint is in $V_j$. 

We construct a text $t$ over alphabet $\Sigma_t$ and a pattern $p$
over alphabet $\Sigma_p$ from $G$ and $k$ such that
the following two conditions hold:
\begin{enumerate}
\item[(C1)] the parameters $\# \Sigma_t$ and $\# \Sigma_p$ are bounded
 by $k$ (note the parameters $\# ?$ and $\max |f(?)|$ are bounded
 since we consider GFM and GPM).
\item[(C2)] $p$ GF/GP-matches $t$ using a function $f$ with $\max_{p
   \in \Sigma_p}|f(p)|\leq 2$ if and only if $G$
 has a $k$-clique.
\end{enumerate}

Let $r=2kn(n-1)+2n+(k-1)m-1$.
The alphabet $\Sigma_t$ consists of (1)
the letter $\gs$, (2)
the letters $l^{i,j}_l$ and $r^{i,j}_l$ for every $1 \leq i <
j \leq k$ and $1 \leq l \leq m-1$, (3)
the letters $l^{v,j}_l$ and $r^{v,j}_l$ for every $v \in V_i$,
$1 \leq j \leq k$, and $1 \leq l \leq n-1$, where
$1 \leq i \leq k$ and $j\neq i$,
(4) the letters $l^{i}_l$ and $r^i_l$ for every $1 \leq i \leq k$
and $1 \leq l \leq r$,
(5) the letter $e_l^{i,j}$ for every $1 \leq i < j \leq k$ and
$1 \leq l \leq m$,
and (6) the letter $\gs_i$ for every $1 \leq i \leq n$.

The alphabet $\Sigma_p$ consists of (1)
the letter $\gs$, (2)
the letters $L^{i,j}_l$ and $R^{i,j}_l$ for every $1 \leq i <
j \leq k$ and $1 \leq l \leq m-1$, 
(3)
the letters $\textup{LL}^{i,j}_l$ and $\textup{RR}^{i,j}_l$ for every $1 \leq i,
j \leq k$ with $i \neq j$, and $1 \leq l \leq n-1$, 
(4) the letters $L^{i}_l$ and $R^i_l$ for every $1 \leq i \leq k$
and $1 \leq l \leq r$,
(5) the letter $E_{i,j}$ for every $1 \leq i < j \leq k$,
and (6) the letter $A_i$ for every $1 \leq i \leq n$.

For a symbol $l$ and $i \in \naturals$, we write
$\enu(l,i)$ to denote the text $l_1 \dotsb l_i$.

Furthermore, for a vertex $v \in V(G)$ and $i$ with $1 \leq i \leq
k$, 
we write $\edges(v,i)$ to
denote the text $\textbf{el}(E_i(v))$,
where $\textbf{el}(E')$ (for a set of edges $E'$)  is a list of all
the letters in $\Sigma_t$ that correspond to the edges in $E'$.

We first define the following preliminary text and pattern
strings. For $i$ and $j$ with $1 \leq i < j \leq k$, we denote by 
$\tij(i,j)$ the text $\enu(l^{i,j},m-1) \enu(e^{i,j},m)
\enu(r^{i,j},m-1)$. We define $t_1$ to be the text:
\begin{quote}
 \begin{center}
   $\gs \tij(1,2) \gs \dotsb \gs \tij(1,k)$

   $\gs \tij(2,3) \gs \dotsb \gs \tij(2,k)$
   
   $\dotsb$

   $\gs \tij(k-1,k) \gs$
 \end{center}
\end{quote}
For a vertex $v \in V_i$, and $j$ with $1 \leq j \leq k$,
we denote by $\tvj(v,j)$ the text
$\enu(l^{v,j},n-1)\edges(v,j) \enu(r^{v,j},n-1)$ if $j\neq
i$ and the empty text if $j=i$. Furthermore, we denote by
$\tv(v)$ the text $\tvj(v,1) \dotsb \tvj(v,k)$.
Let $t_2$ be the text:
\begin{quote}
 \begin{center}
   $\enu(l^1,r) \gs_1 \tv(v_1^1) \gs_1 \dotsb \gs_n \tv(v_n^1)
   \gs_n \enu(r^1,r)$
   
   $\dotsb$
   
   $\gs \enu(l^k,r) \gs_1 \tv(v_1^k) \gs_1 \dotsb \gs_n \tv(v_n^k) \gs_n \enu(r^k,r)$
 \end{center}
\end{quote}

For $i$ and $j$ with $1 \leq i < j \leq k$, we denote by 
$\pij(i,j)$ the pattern $\enu(L^{i,j},m-1) E_{i,j} \enu(R^{i,j},m-1)$.
Let $p_1$ be the pattern:
\begin{quote}
 \begin{center}
   
   $\gs \pij(1,2) \gs \dotsc \gs \pij(1,k)$
   
   $\gs \pij(2,3) \gs \dotsc \gs \pij(2,k)$
   
   $\dotsb$
   
   $\gs \pij(k-1,k) \gs$
 \end{center}
\end{quote}
For $i$, $j$ with $1 \leq i,j \leq k$, let $I(i,j)$ be 
the letter $E_{i,j}$ if $i < j$, the letter $E_{j,i}$ if $i>j$ and the empty string if
$i=j$. Furthermore, let $\pijs(i,j)$ be the pattern 
$\enu(\textup{LL}^{i,j},n-1) I(i,j) \enu(\textup{RR}^{i,j},n-1)$ if $i \neq j$ and the
empty pattern otherwise.
Let $p_2$ be the pattern:
\begin{quote}
 \begin{center}
   $\enu(L^1,r) A_1 \pijs(1,1) \dotsb \pijs(1,k) A_1 \enu(R^1,r)$
   
   $\dotsb$
   
   $\gs \enu(L^k,r) A_k \pijs(k,1) \dotsb \pijs(k,k) A_k \enu(R^k,r)$
 \end{center}
\end{quote}
We also define $t_0$ to be the text $\gs\gs$ and $p_0$ to be
the pattern $\gs\gs$.
Then, $t$ is the concatenation of $t_0$, $t_1$ and $t_2$
and $p$ is a concatenation of $p_0$, $p_1$ and $p_2$.

This completes the construction of $t$ and $p$. Clearly, $t$ and $p$
can be constructed from $G$ and $k$ in fpt-time (even polynomial
time). Furthermore, because $\# \Sigma_t=\binom{k}{2}+k+2$, 
$|\# \Sigma_p|=\binom{k}{2}+k+2$, condition (C1) is satisfied.
To show the remaining condition (C2) we need the following
intermediate lemmas and claims. 

\begin{LEM}\label{lem:GFM-GPM-occ-1}
 If $G$ has a $k$-clique then $p$ GF/GP-matches $t$ using a function
 $f$ with $\max_{p \in \Sigma_p}|f(p)|=2$.
\end{LEM}
\begin{proof}
 Let $\{v_{h_1}^1, \dotsc,v_{h_k}^k\}$ be the vertices and 
 $\SB e^{i,j}_{h_{i,j}} \SM 1 \leq i < j \leq k \SE$ be the edges of
 a $k$-clique of $G$ with $1 \leq h_j \leq n$  and $1 \leq h_{i,j}
 \leq m$ for every $i$ and $j$ with $1 \leq i < j \leq k$. 

 We first give the GF/GP-matching function $f$ for the letters in $\Sigma_p$ that
 occur more than once in $p$ as follows:
 We set $f(\gs)=\gs$, 
 $f(E_{i,j})=e^{i,i}_{h_{i,j}}$, and
 $f(A_i)=\#_{h_i}$,
 for every $i$ and $j$ with $1 \leq i < j \leq k$. 
 Informally, we will map the remaining
 letters in $\Sigma_p$ to substrings of $t$ of length
 between $1$ and $2$ in such a way that the occurrences of the
 letters $\gs$, $E_{i,j}$, and $A_i$ are placed over the right
 positions in the text $t$.
 More formally, we define $f$ for the remaining letters in
 $\Sigma_p$ as follows:
 \begin{itemize}
 \item For every $1 \leq i < j \leq k$, we define
   $f(L_l^{i,j})$ in such a
   way that $|f(L_l^{i,j})|=2$ for every $1 \leq l \leq h_{i,j}-1$
   and $|f(L_l^{i,j})|=1$ for every $h_{i,j}-1 < l \leq m-1$.
 \item For every $1 \leq i < j \leq k$, we define
   $f(R_l^{i,j})$ in such a
   way that $|f(R_l^{i,j})|=1$ for every $1 \leq l \leq h_{i,j}+1$
   and $|f(L_l^{i,j})|=2$ for every $h_{i,j}+1 < l \leq m-1$.
 \item For every $1 \leq i,j \leq k$ with $i \neq j$, we define
   $f(\textup{LL}^{i,j}_l)$ in such a way that
   $f(\textup{LL}^{i,j}_l)=2$ for every $1 \leq l \leq s-1$, where $s$
   is the position of $e^{i,j}_{h_{i,j}}$ in $t(v_{h_i},j)$ and
   $f(\textup{LL}^{i,j}_l)=1$ for every $s < l \leq n-1$.
 \item For every $1 \leq i,j \leq k$ with $i \neq j$, we define
   $f(\textup{RR}^{i,j}_l)$ in such a way that
   $f(\textup{RR}^{i,j}_l)=1$ for every $1 \leq l \leq s+1$, where $s$
   is the position of $e^{i,j}_{h_{i,j}}$ in $t(v_{h_i},j)$ and
   $f(\textup{RR}^{i,j}_l)=1$ for every $s+1 < l \leq n-1$.
 \item For every $1 \leq i \leq k$, we define
   $f(L_l^{i})$ in such a
   way that $|f(L_l^{i})|=2$ for every $1 \leq l \leq s-1$, where 
   $s$ is position of $\gs_{h_i}$ in the substring
   $\gs_1 \tv(v_1^i) \gs_1 \dotsb \gs_n \tv(v_n^i)
   \gs_n$ of $t$
   and $|f(L_l^{i,j})|=1$ for every $s < l \leq r$.
 \item For every $1 \leq i \leq k$, we define
   $f(R_l^{i})$ in such a
   way that $|f(R_l^{i})|=1$ for every $1 \leq l \leq s+1$, where 
   $s$ is position of $\gs_{h_i}$ in the substring
   $\gs_1 \tv(v_1^i) \gs_1 \dotsb \gs_n \tv(v_n^i)
   \gs_n$ of $t$
   and $|f(R_l^{i,j})|=2$ for every $s+1 < l \leq r$.
 \end{itemize}
 It is now straightforward to check that $f$ GF/GP-matches $p$ to $t$
 and $\max_{p \in \Sigma_p}|f(p)|=2$, as required.
\end{proof}

To prove the reverse direction we need the following
intermediate claim.
\begin{CLM}\label{clm:GFM-occt-max-1}
 For any function $f$ that $k'$-GF/GP-matches $p$ to $t$ it holds that:
 $f(\es)=\es$, $f(\gs)=\gs$, and $f(Q)=+$.
 Moreover, all wildcards have to be placed on all the $k'$
 occurrences of $Q$ in $p$.
\end{CLM}
\begin{proof}
 We first show that $f(Q)=+$. Observe that the concatenation of the
 strings $t_1$ and $t_2$ is
 square-free (recall the definition of square-free from
 Section~\ref{sec:prel}). It follows that every two consecutive
 occurrences of pattern letters in $p_0$ have to be mapped to a
 substring of $t_0$. Because there are $2(k'+1)$ occurrences of $Q$
 in $p_0$ it follows that at least two consecutive occurrences of
 $Q$ in $p_0$ are not replaced with wildcards and hence $Q$ has to be
 mapped to a substring of $t_0$. Furthermore, since all occurrences
 of $Q$ are at the end of $p_0$, we obtain that $Q$ has to be mapped
 to $+$, as required. Because all occurrences of $Q$ in $p_1$ and $p_2$ have to
 be mapped to substrings of the concatenation of $t_1$ and $t_2$
 but neither $t_1$ nor $t_2$ contain the
 letter $+$, it follows that all the $k'$ occurrences of $Q$ in
 $p_1$ and $p_2$
 have to be replaced by wildcards. Since we are only allowed to use
 at most $k'$ wildcards, this shows the second statement of the
 claim. Since no wildcards are used to replace letters in $p_0$ it
 now also follows that $f(\es)=\es$ and $f(\gs)=\gs$.
\end{proof}
\begin{LEM}\label{lem:GFM-occt-max-2}
 If $p$ GF/GP-matches $t$ using a a function $f$ with $\max_{p \in
   \Sigma_p}|f(p)|=2$, then $G$ has
 a $k$-clique.
\end{LEM}
\begin{proof}
 Let $f$ be the function that GF/GP-matches $p$ to $t$ with 
 $\max_{p \in \Sigma_p}|f(p)|=2$. We first
 show that $f(\gs)=\gs$. 
 Suppose for a contradiction that $f(\gs)\neq \gs$
 Because $t$ and $p$ start with $\gs\gs$ it follows that $f(\gs)$
 is a string that starts with $\gs\gs$. However, $t$ does not
 contain any other occurrence of the string $\gs\gs$ and hence the remaining
 occurrences of $\gs$ in $p$ cannot be matched by $f$. 

 Because $t$ and $p$ have the same number of occurrences of $\gs$
 , it follows that the $i$-th occurrences of $\gs$ in $p$ 
 has to be mapped to the $i$-th occurrence of $\gs$
 in $t$. We obtain that:
 \begin{enumerate}
 \item[(1)] for every $i$, $j$ with $1 \leq i < j \leq k$, the substring
   $\pij(i,j)$ of $p$ has to be mapped to the substring
   $\tij(i,j)$ of $t$.
 \item[(2)] for every $i$ with $1 \leq i \leq k$,  the substring
   $\enu(L^i,r) A_i \pijs(i,1) \dotsb \pijs(i,k) A_i \enu(R^i,r)$
   of $p$ has to be mapped to the substring
   $\enu(l^i,r) \gs_1 \tv(v_1^i) \gs_1 \dotsb \gs_n \tv(v_n^i)
   \gs_n \enu(r^i,r)$ of $t$.
 \end{enumerate}
 Because for every $i$ with $1 \leq i \leq k$ the letters $\gs_j$
 are the only letters that occur more than once in the substring 
 $\enu(l^i,r) \gs_1 \tv(v_1^i) \gs_1 \dotsb \gs_n \tv(v_n^i)
 \gs_n \enu(r^i,r)$
 of $t$, we obtain that $A_i$ has to be mapped to $\gs_j$ for some
 $j$ with $1 \leq j \leq n$. Consequently:
 \begin{enumerate}
 \item[(3)] for every $i$ with $1 \leq i \leq k$,  the substring
   $A_i \pijs(i,1) \dotsb \pijs(i,k) A_i$
   of $p$
   has to be mapped a substring
   $\gs_j \tv(v_j^i) \gs_j$ of $t$ for some $j$ with $1 \leq j
   \leq n$.
 \end{enumerate}
 It follows from (1) that for every $i$, $j$ with $1 \leq i < j
 \leq k$, the function $f$ maps $E_{i,j}$ to edges between
 $V_i$ and $V_j$. W.l.o.g. we can assume that $E_{i,j}$ is mapped
 to exactly one such edge because mapping it to many edges makes it
 only harder to map the following occurrences of $E_{i,j}$ in $p$.
 Because of (3) it follows that for every $i$
 with $1 \leq i \leq k$, it holds that the edges mapped to any
 $E_{l,r}$ with $1 \leq l < r \leq k$ such that $l=i$ or $r=i$ have
 the same endpoint in $V_i$. Hence, the set of edges mapped to all
 the letters $E_{i,j}$ for $1 \leq i <j \leq k$ form a $k$-clique
 of $G$. 
\end{proof}
This concludes the proof of Theorem~\ref{the:GFM-occt-max}.

\begin{THE}\label{the:GFM-questionmarksize-hard}
  Max-GFM and Max-GPM are \W{1}\hy hard parameterized by $\#
  \Sigma_p$, $|\Sigma_p|$, $\max_i |f(p_i)|$, and $\# ?$ (even if $\max_i |f(p_i)|=1$).
\end{THE}
We will show the above theorem by a  parameterized reduction from \textsc{Multicolored
  Clique}. Let $G = (V, E)$ be a $k$-partite graph with partition
$V_1,\dots,V_k$ of $V$. Let $E_{i,j}=\SB \{u,v\} \in E \SM  u \in
V_i \textup{ and } v \in V_j \SE$ for every $i$ and $j$ with $1 \leq
i < j \leq k$. 
As we stated in the preliminaries we can assume that $|V_i|=n$ and
$|E_{i,j}|=m$ for every $i$ and $j$ with $1 \leq i < j \leq k$.

Let $V_i=\{v_1^i,\dotsc,v_n^i\}$,
$E_{i,j}=\{e_1^{i,j},\dotsc,e_m^{i,j}\}$, and $k'=2\binom{k}{2}$. 
We construct a text $t$ over alphabet $\Sigma_t$ and a pattern $p$
over alphabet $\Sigma_p$ from $G$ and $k$ in polynomial time such that:
\begin{enumerate}
\item[(C1)] the parameters $\# \Sigma_p$, $|\Sigma_p|$, and $\# ?$ can be
  bounded as a function of $k$.
\item[(C2)] $p$ $k'$-GF/GP-matches $t$ using a function $f$ with
  $\max_{p \in \Sigma_p}|f(p)|=1$ if and only if $G$ has a $k$-clique.
\end{enumerate}
We set
$\Sigma_t=\{\es,\ies,\gs,+\} \cup \SB l_{i,j},r_{i,j} \SM 1 \leq i < j
\leq k \SE \cup \SB v_i^j \SM 1 \leq i \leq n
\textup{ and }1 \leq j \leq k \SE$ and 
$\Sigma_p=\{\es,\ies,\gs, D\} \cup \SB V_i \SM 1 \leq i \leq k \SE$.

For an edge $e \in E$ between $v_l^i$ and $v_s^j$ where $1 \leq i
< j \leq k$ and $1 \leq l,s \leq n$, we write $\vertt(e)$ to
denote the text $v_l^i \ies v_s^j$. For $l \in \Sigma_p \cup \Sigma_t$ and $i \in
\naturals$ we write $\rep(l,i)$ to denote the text consisting of
repeating the letter $l$ exactly $i$ times.
We first define a preliminary text $t'$ as follows.
\begin{quote}
  \begin{center}
    $\gs l_{1,2}\es \vertt(e_1^{1,2}) \es \dotsb \es \vertt(e_m^{1,2})
    \es r_{1,2} \gs \dotsb \gs l_{1,k} \es \vertt(e_1^{1,k}) \es \dotsb \es
    \vertt(e_m^{1,k}) \es r_{1,k}$
    
    $\gs l_{2,3} \es \vertt(e_1^{2,3}) \es \dotsb \es \vertt(e_m^{2,3})
    \es r_{2,3} \gs \dotsb \gs l_{2,k}\es \vertt(e_1^{2,k}) \es \dotsb \es
    \vertt(e_m^{2,k}) \es r_{2,k}$
    
    $\dotsb$
    
    $\gs l_{k-1,k} \es \vertt(e_1^{k-1,k}) \es \dotsb \es
    \vertt(e_m^{k-1,k}) \es r_{k-1,k}\gs$
  \end{center}
\end{quote}
We also define a preliminary pattern $p'$ as follows.
\begin{quote}
  \begin{center}
    
    $\gs D \es V_1 \ies V_2 \es D \gs \dotsc \gs D \es V_1 \ies V_k
    \es D$
    
    $\gs D \es V_2 \ies V_3 \es D \gs \dotsc \gs D \es V_2 \ies V_k
    \es D$
    
    $\dotsb$
    
    $\gs D \es V_{k-1} \ies V_k \es D \gs$
  \end{center}
\end{quote}
Let $r=2(k'+1)$. Then $t$ is obtained by preceding $t'$ with the
text $t''$ defined as follows.
\begin{quote}
  \begin{center}
    $\gs\es\ies\rep(+,r)$
  \end{center}
\end{quote}
Similarly, $p$ is obtained by preceding $p'$ with the
text $p''$ defined as follows.
\begin{quote}
  \begin{center}
    $\gs\es\ies\rep(D,r)$
  \end{center}
\end{quote}
This completes the construction of $t$ and $p$. Clearly, $t$ and $p$
can be constructed from $G$ and $k$ in fpt-time (even polynomial
time). Furthermore, because $\# \Sigma_p=r+k'=2(k'+1)+k'=3k'+1$,
$|\Sigma_p|=k+4$, and $\# ?=k'$, condition (C1) above is satisfied.
To show the remaining condition (C2),
we need the following intermediate lemmas.
\begin{LEM}\label{lem:GFM-questionmarksize-hard-1}
  If $G$ has a $k$-clique, then $p$ $k'$-GF/GP-matches to $t$ using
  a function $f$ with $\max_{p \in \Sigma_p}|f(p)|=1$.
\end{LEM}
\begin{proof}
  Let $\{v_{h_1}^1, \dotsc,v_{h_k}^k\}$ be the vertices and 
  $\SB e^{i,j}_{h_{i,j}} \SM 1 \leq i < j \leq k \SE$ be the edges of
  a $k$-clique of $G$ with $1 \leq h_j \leq n$  and $1 \leq h_{i,j}
  \leq m$ for every $i$ and $j$ with $1 \leq i < j \leq k$. 

  We put $k'$ wildcards on the last $k'$ occurrences of $D$ in $p$.  
  Informally, these wildcards are mapped in such a way that for every
  $1 \leq i < j \leq k$
  the substring $\es V_i \ies V_j \es$ of the pattern $p$ is mapped 
  to the substring $\es \vertt(e_{h_{i,j}}^{i,j}) \es$ of the text $t$.
  More formally, for $i$ and $j$ with $1 \leq i < j \leq k$ let
  $q=(\sum_{o=1}^{o<i}(k-o))+j$. We
  map the wildcard on the $2(q-1)$-th occurrence of the letter $D$ in $p'$
  with the text $l_{i,j} \es \vertt(e^{i,j}_1) \es \dotsb \es
  \vertt(e^{i,j}_{h_{i,j}-1})$ 
  and similarly we map the wildcard on the $(2(q-1)+1)$-th
  occurrence of the letter $D$ in $p'$ with the text
  $ \vertt(e^{i,j}_{h_{i,j}+1}) \es \dotsb \es
  \vertt(e^{i,j}_{m})\es r_{i,j}$. Note that in this way every
  wildcard is mapped to a non-empty substring of $t$ and no two
  wildcards are mapped to the same substring of $t$, as required.

  We then define the $k'$-GF/GP-matching function $f$ as follows: 
  $f(\es)=\es$, $f(\ies)=\ies$, $f(\gs)=\gs$, 
  $f(V_i)=v_{h_i}^i$, 
  $f(D)=+$,
  for every $i$ and $h_i$ with $1 \leq i \leq k$ and $1 \leq h_i \leq
  n$. 
  It is straightforward to check that $f$ together with the mapping for the
  wildcards maps the pattern $p$ to the text $t$. 
\end{proof}
\begin{LEM}\label{clm:GFM-questionmarksize-hard-s21}
  Let $f$ be a function that $k'$-GF/GP-matches $p$ to $t$ with
  $\max_{p \in \Sigma_p}|f(p)|=1$, then:
  $f(\es)=\es$, $f(\ies)=\ies$, $f(\gs)=\gs$, and $f(D)=+$.
  Moreover, all wildcards have to be placed on all the $k'$
  occurrences of $D$ in $p'$.
\end{LEM}
\begin{proof}
  We first show that $f(D)=+$. Observe that the string $t'$ is
  square-free (recall the definition of square-free from
  Section~\ref{sec:prel}). It follows that every two consecutive
  occurrences of pattern letters in $p''$ have to be mapped to a
  substring of $t''$. Because there are $2(k'+1)$ occurrences of $D$
  in $p''$ it follows that at least two consecutive occurrences of
  $D$ in $p''$ are not replaced with wildcards and hence $D$ has to be
  mapped to a substring of $t''$. Furthermore, since all occurrences
  of $D$ are at the end of $p''$, we obtain that $D$ has to be mapped
  to $+$, as required. Because all occurrences of $D$ in $p'$ have to
  be mapped to substrings of $t'$ and $t'$ does not contain the
  letter $+$, it follows that all the $k'$ occurrences of $D$ in $p'$
  have to be replaced by wildcards. Since we are only allowed to use
  at most $k'$ wildcards, this shows the second statement of the
  lemma. Since no wildcards are used to replace letters in $p''$ it
  now easily follows that $f(\es)=\es$, $f(\ies)=\ies$ and $f(\gs)=\gs$.
\end{proof}
\begin{LEM}\label{lem:GFM-questionmarksize-hard-2}
  If $p$ $k'$-GF/GP-matches to $t$ using
  a function $f$ with $\max_{p \in \Sigma_p}|f(p)|=1$, then
  $G$ has a $k$-clique.
\end{LEM}
\begin{proof}\sloppypar
  Let $f$ be a function that $k'$-GF/GP-matches $p$ to
  $t$ such that $\max_{p \in \Sigma_p}|f(p)|=1$. 
  We claim that the set $\SB f(V_i) \SM 1 \leq i \leq k \SE$ is a
  $k$-clique of $G$. Because of
  Lemma~\ref{clm:GFM-questionmarksize-hard-s21}, we know that
  $f(\gs)=\gs$ and that no
  occurrence of $\gs$ in $p$ is replaced by a wildcard. Since the
  number of occurrences of $\gs$ in $t$ is equal to the number of
  occurrences of $\gs$ in $p$, we obtain that the $i$-th occurrence of
  $\gs$ in $p$ is mapped to the $i$-th occurrence of $\gs$ in
  $t$. Consequently, for every $i$ and $j$ with $1 \leq i < j \leq k$,
  we obtain that the substring $\es V_i\ies V_j\es$ is mapped to a substring
  of the string $l_{i,j} \es \vertt(e_1^{i,j}) \es \dotsb \es
  \vertt(e_m^{i,j}) \es r_{i,j}$ in $t$. Again, using
  Lemma~\ref{clm:GFM-questionmarksize-hard-s21}
  and the fact that $\max_{p \in \Sigma_p}|f(p)|=1$, 
  we obtain that both $V_i$ and $V_j$ are mapped
  to some letter $v^i_l$ and $v_s^j$ for some $l$ and $s$ with $1 \leq
  l,s\leq n$ such that $\{v^i_l,v^j_s\} \in E$. Hence, 
  $\SB f(V_i) \SM 1 \leq i \leq k \SE$ is a $k$-clique of $G$. 
\end{proof}
Because Condition (C2) is implied by Lemmas~\ref{lem:GFM-questionmarksize-hard-1}
and~\ref{lem:GFM-questionmarksize-hard-2}, this concludes
the proof of Theorem~\ref{the:GFM-questionmarksize-hard}.

\begin{THE}\label{the:GFM-mobile-1}
  (Max-)GFM and (Max-)GPM are \W{1}\hy hard parameterized by $\#\Sigma_t$,
  $\# \Sigma_p$, $|\Sigma_p|$, $\# ?$, and $\max |f(?)|$.
\end{THE}

We will show the theorem by a  parameterized reduction from \textsc{Multicolored
  Clique}. Let $G = (V, E)$ be a $k$-partite graph with partition
$V_1,\dots,V_k$ of $V$. Let $E_{i,j}=\SB \{u,v\} \in E \SM  u \in
V_i \textup{ and } v \in V_j \SE$ for every $i$ and $j$ with $1 \leq
i < j \leq k$. 
Again, as we stated in the preliminaries we can assume that $|V_i|=n$ and
$|E_{i,j}|=m$ for every $i$ and $j$ with $1 \leq i < j \leq k$.

Let $V_i=\{v_1^i,\dotsc,v_n^i\}$ and $E_{i,j}=\{e_1^{i,j},\dotsc,e_m^{i,j}\}$.
We construct a text $t$ and a pattern $p$ from $G$ and $k$ such that
$p$ GF/GP-matches $t$ if and only if $G$ has a $k$-clique. 
The alphabet $\Sigma_t$ consists of:
\begin{itemize}
\item the letter $\gs$ (used as a separator);
\item one letter $a_e$ for every $e \in E$ (representing the
  edges of $G$);
\item one letter $\gs_i$ for every $i$ with $1 \leq i \leq n$ (used as
  special separators that group edges from the same vertex);
\item the letters $l_{i,j}$, $r_{i,j}$, $l_i$, $r_i$ for every $i$ and $j$ with
  $1 \leq i < j \leq k$ (used as dummy letters to ensure injectivity
  for GPM);
\item the letter $d_e^v$ and $d^v$ for every $e \in E$ and $v \in V(G)$ with
  $v \in e$ (used as dummy letters to ensure injectivity for GPM).
\end{itemize}
We set $\Sigma_p=\{\gs\} \cup \SB E_{i,j}, L_{i,j}, R_{i,j},L_i,R_i,A_i \SM 1 \leq i < j \leq k
\SE \cup \SB D_{i,j} \SM 1 \leq i \leq k \textup{ and } 1 \leq j \leq k+1 \SE$.

For a vertex $v \in V$ and $j$ with $1 \leq j
\leq k$ we denote by $E_j(v)$ the set of edges of $G$ that are
incident to $v$ and whose other endpoint is in $V_j$. 
Furthermore, for a vertex $v \in V(G)$, we write $\edges(v)$ to
denote the text $\textbf{el}(v,E_1(v))\dotsb \textbf{el}(v,E_k(v))d^v$,
where $\textbf{el}(v,E')$, for vertex $v$ and a set $E'$ of edges with $E'=\{e_1,\dotsc,e_l\}$, 
is the text $d_{e_1}^v a_{e_1} d_{e_2}^v a_{e_2} \dotsb d_{e_l}^v
a_{e_l}$.

We first define the following preliminary text and pattern
strings. Let $t_1$ be the text:
\begin{quote}
  \begin{center}
    $\gs l_{1,2} a_{e^{1,2}_1} \dotsb a_{e^{1,2}_m} r_{1,2} \gs \dotsb \gs l_{1,k} a_{e^{1,k}_1} \dotsb a_{e^{1,k}_m} r_{1,k}$
    
    $\gs l_{2,3} a_{e^{2,3}_1} \dotsb a_{e^{2,3}_m} r_{2,3} \gs \dotsb \gs l_{2,k} a_{e^{2,k}_1} \dotsb a_{e^{2,k}_m} r_{2,k}$
    
    $\dotsb$
    
    $\gs l_{k-1,k} a_{e^{k-1,k}_1} \dotsb a_{e^{k-1,k}_m} r_{k-1,k}$
  \end{center}
\end{quote}
Let $t_2$ be the text:
\begin{quote}
  \begin{center}
    $\gs l_1 \gs_1 \edges(v_1^1) \gs_1 \dotsb \gs_n \edges(v_n^1) \gs_n
    r_1$
    
    $\dotsb$
    
    $\gs l_k \gs_1 \edges(v_1^k) \gs_1 \dotsb \gs_n \edges(v_n^k)
    \gs_n r_k \gs$
  \end{center}
\end{quote}
Let $p_1$ be the pattern:
\begin{quote}
  \begin{center}
    
    $\gs L_{1,2} E_{1,2} R_{1,2} \gs \dotsc \gs L_{1,k} E_{1,k} R_{1,k}$
    
    $\gs L_{2,3} E_{2,3} R_{2,3} \gs \dotsc \gs L_{2,k} E_{2,k} R_{2,k}$
    
    $\dotsb$
    
    $\gs L_{k-1,k} E_{k-1,k} R_{k-1,k}$
  \end{center}
\end{quote}


For $i$, $j$ with $1 \leq i,j \leq k$, let $I(i,j)$ be 
the letter $E_{i,j}$ if $i < j$, the letter $E_{j,i}$ if $i>j$ and the empty string if
$i=j$. We define $\pih(1)$ to be the
pattern:
\begin{quote}
  \begin{center}
    $A_1 D_{1,2} I(1,2) D_{1,3} I(1,3) \dotsb \dotsb D_{1,k}I(1,k) D_{1,k+1} A_1$
  \end{center}
\end{quote}
we define $\pih(k)$ to be the
pattern:
\begin{quote}
  \begin{center}
    $A_k D_{k,1} I(k,1) D_{k,2} I(k,2) \dotsb \dotsb D_{k,k-1} I(k,k-1) D_{k,k+1} A_k$
  \end{center}
\end{quote}
and for every $i$ with $1 < i < k$, we define
$\pih(i)$ to be the pattern:
\begin{quote}
  \begin{center}
    $A_i D_{i,1} I(i,1) D_{i,2} I(i,2) \dotsb
    D_{i,i-1}I(i,i-1)$

    $D_{i,i+1}I(i,i+1) \dotsb D_{i,k} I(i,k) D_{i,k+1}
    A_i$
  \end{center}
\end{quote}
Then $p_2$ is the pattern:
\begin{quote}
  \begin{center}

    $\gs L_1 \pih(1) R_1 \gs \dotsb \gs L_k \pih(k) R_k \gs$
  \end{center}
\end{quote}
We also define $t_0$ to be the text $\gs\gs$ and $p_0$ to be
the pattern $\gs\gs$.
Then, $t$ is the concatenation of $t_0$, $t_1$ and $t_2$
and $p$ is a concatenation of $p_0$, $p_1$ and $p_2$.

This completes the construction of $t$ and $p$. Clearly, $t$ and $p$
can be constructed from $G$ and $k$ in fpt-time (even polynomial
time). Furthermore, $\# \Sigma_t=\binom{k}{2}+k+3$, 
$\# \Sigma_p=\binom{k}{2}+k+3$,
$|\Sigma_p|=k(k+1)+3\binom{k}{2}+3k+1$
and hence bounded by $k$,
as required. It remains to show that $G$ has a $k$-clique if and only
if $p$ GF/GP-matches $t$. 

\begin{LEM}\label{lem:GFM-mobile-1-hard-1}
  If $G$ has a $k$-clique then $p$ GF/GP-matches $t$.
\end{LEM}
\begin{proof}\sloppypar
  Let $\{v_{h_1}^1, \dotsc,v_{h_k}^k\}$ be the vertices and 
  $\SB e_{h_{i,j}}^{i,j} \SM 1 \leq i < j \leq k \SE$ be the edges of
  a $k$-clique of $G$ with $1 \leq h_i \leq n$ and $1 \leq h_{i,j}
  \leq m$ for every $i$ and $j$ with $1 \leq i < j \leq k$. 

  We define the function $f$ that GF/GP-matches $p$ to $t$ as follows:
  We set $f(\gs)=\gs$ and $f(\es)=\es$. Moreover, for every $i$ and
  $j$ with $1 \leq i < j \leq k$, we set
  $f(E_{i,j})=a_{e_{h_{i,j}}^{i,j}}$, $f(A_i)=\gs_{i}$,
  $f(L_{i,j})=l_{i,j} a_{e^{i,j}_1} \dotsb a_{e^{i,j}_{h_{i,j}-1}}$,
  $f(R_{i,j})=a_{e^{i,j}_{h_{i,j}+1}}\dotsb a_{e^{i,j}_m}r_{i,j}$,
  $f(L_i)=\es l_i \gs_1 \edges(v_1^i) \gs_1 \dotsb \gs_{h_i-1}
    \edges(v_{h_i-1}^i) \gs_{h_i-1}$, and
  $f(R_i)=\gs_{h_i+1} \edges(v_{h_i+1}^i) \gs_{h_i+1} \dotsb
    \gs_{n} \edges(v_{n}^i) \gs_{n} r_i$.

  For every $i$ and $j$ with $i \neq j$, let $e(i,j)$ be the
  edge $e_{h_{i,j}}^{i,j}$ if $i < j$ and the edge
  $e_{h_{j,i}}^{j,i}$, otherwise. Then, the letters $D_{i,j}$ are
  mapped as follows:
  \begin{itemize}
  \item For every $i$ and $j$ with $1
    \leq i \leq k$, $2 \leq j \leq k$, $i \neq j$, and $(i,j)\neq
    (1,2)$,
    we map $f(D_{i,j})$ to the substring of $\edges(v_{h_i}^i)$
    in between the occurrences (and not including these occurrences) 
    of the letters $e(i,j-1)$ and $e(i,j)$.
  \item We map $f(D_{1,2})$ to be the prefix
    of $\edges(v_{h_1}^1)$ ending before the letter $e(1,2)$.
  \item For every $i$ with $2 \leq i \leq k$, we map
    $f(D_{i,1})$ to the prefix of $\edges(v_{h_i}^i)$
    ending before the letter $e(i,1)$.
  \item For every $i$ with $1 \leq i < k$, we map
    $f(D_{i,k+1})$ to the suffix of $\edges(v_{h_i}^i)$
    starting after the letter $e(i,k)$.
  \item We map $f(D_{k,k+1})$ to be the suffix
    of $\edges(v_{h_1}^1)$ starting after the letter $e(k,k-1)$.
  \end{itemize}
  It is straightforward to check that $f$ GF/GP-matches $p$ to $t$.
\end{proof}

\begin{LEM}\label{lem:GFM-mobile-1-hard-1}
  If $p$ GF/GP-matches $t$ then $G$ has a $k$-clique.
\end{LEM}
\begin{proof}
  Let $f$ be a function that GF/GP-matches $p$ to $t$. We first
  show that $f(\gs)=\gs$. 
  Suppose for a contradiction that $f(\gs)\neq \gs$.
  Because $t$ and $p$ start with $\gs\gs$ it follows that $f(\gs)$
  is a string that starts with $\gs\gs$. However, $t$ does not
  contain any other occurrence of the string $\gs\gs$ and hence the remaining
  occurrences of $\gs$ in $p$ cannot be matched by $f$.

  Because $t$ and $p$ have the same number of occurrences of $\gs$,
  it follows that the $i$-th occurrences of $\gs$ in $p$ has to be
  mapped to the $i$-th occurrence of $\gs$ in $t$. We obtain that:
  \begin{enumerate}
  \item[(1)] For every $i$, $j$ with $1 \leq i < j \leq k$, the substring
    $L_{i,j} E_{i,j} R_{i,j}$ of $p$ has to be mapped to the substring
    $l_{i,j} a_{e^{i,j}_1} \dotsb a_{e^{i,j}_m} r_{i,j}$ of $t$.
  \item[(2)] For every $i$ with $1 \leq i \leq k$,  the substring
    $L_i \pih(i) R_i$ of $p$ has to be mapped to the substring
    $l_i \gs_1 \edges(v_1^i) \gs_1 \dotsb \gs_n \edges(v_n^i) \gs_n r_i$ of $t$.
  \end{enumerate}
  Because for every $i$ with $1 \leq i \leq k$ the letters $\gs_j$
  are the only letters that occur more than once in the substring 
  $l_i \gs_1 \edges(v_1^i) \gs_1 \dotsb \gs_n \edges(v_n^i) \gs_n r_i$
  of $t$, we obtain from (2) that $A_i$ has to be mapped to $\gs_j$ for some
  $j$ with $1 \leq j \leq n$. Consequently:
  \begin{enumerate}
  \item[(3)] for every $i$ with $1 \leq i \leq k$,  the substring
    $\pih(i)$ of $p$
    has to be mapped to a substring
    $\gs_j \edges(v_j^i) \gs_j$ of $t$ for some $j$ with $1 \leq j
    \leq n$.
  \end{enumerate}
  It follows from (1) that for every $i$, $j$ with $1 \leq i < j
  \leq k$, $f(E_{i,j})$ is mapped to some edges between $V_i$ and
  $V_j$. W.l.o.g. we can assume that $f(E_{i,j})$ is mapped to exactly
  one edge between $V_i$ and $V_j$, because mapping it to more than
  one edge would make the matching of the latter occurrences of
  $E_{i,j}$ in $p$ even harder, i.e., whenever a latter occurrence of
  $E_{i,j}$ in $p$ can be mapped to more than one edge it also can be
  mapped to any of these edges. 
  Furthermore, because of (3) it follows that for every $i$
  with $1 \leq i \leq k$, it holds that the edges mapped to any
  $E_{l,r}$ with $1 \leq l < r \leq k$ such that $l=i$ or $r=i$ have
  the same endpoint in $V_i$. Hence, the set of edges mapped to 
  the letters $E_{i,j}$ for $1 \leq i <j \leq k$ form a $k$-clique
  of $G$.
\end{proof}
This concludes the proof of Theorem~\ref{the:GFM-mobile-1}.

\end{document}